\documentclass[11pt]{article}

\usepackage{amsmath,amssymb,amsthm,bbm,graphicx,color}
\usepackage{a4wide}

\newtheorem{definition}{Definition}
\newtheorem{lemma}{Lemma}
\newtheorem{theorem}{Theorem}

\begin{document}

\title{Speed   scaling  with  power   down  scheduling   \\
for  agreeable deadlines\thanks{This work  is supported  by the ANR  grants NETOC
    and TODO as well as by the GdR recherche op\'erationnelle.}}
\date{}

\author{Evripidis Bampis\thanks{LIP6, Universit\'e Pierre et Marie Curie, Paris, France}
\and
Christoph D\"urr\footnotemark[2]~\thanks{CNRS}
\and
Fadi Kacem\thanks{IBISC, Universit\'e d'Evry, France}
\and
Ioannis Milis\thanks{Department of Informatics, Athens University of Economics and Business, Greece}
}

\maketitle

\begin{abstract}
  We consider the problem of  scheduling on a single processor a given
  set of  $n$ jobs. Each  job $j$ has  a workload $w_j$ and  a release
  time $r_j$. The processor can vary its speed and hibernate to reduce
  energy  consumption.   In  a  schedule minimizing  overall  consumed
  energy, it  might be  that some jobs  complete arbitrarily  far from
  their  release  time.  So  in  order  to  guarantee some  quality  of
  service, we  would like to  impose a deadline $d_j=r_j+F$  for every
  job  $j$, where  $F$  is a  guarantee  on the  \emph{flow time}.   We
  provide  an  $O(n^3)$  algorithm   for  the  more  general  case  of
  \emph{agreeable  deadlines},  where  jobs  have  release  times  and
  deadlines  and  can be  ordered  such  that  for every  $i<j$,  both
  $r_i\leq r_j$ and $d_i\leq d_j$.
\end{abstract}


\section{Introduction}

Recent research addresses  the issue of reducing the  amount of energy
consumed by  computer systems while maintaining  satisfactory level of
performance.   This can  be done  at  different levels  of a  computer
system.  One possibility is to  specify a good scheduling mechanism in
the operating system level. Here  we have two mechanisms at hand.  One
common method  for saving  energy is the  \emph{power-down mechanism},
which is to  simply suspend the system during  long enough idle times.
Another common method is \emph{speed  scaling}, which is to adjust the
processor  speed low  enough to  meet the  jobs requirements.  In this
paper  we study  the problem  of designing  scheduling  algorithms for
minimizing the consumed energy using both mechanisms.

The question whether this problem can be solved in polynomial time was
posed by Irani and Pruhs~\cite{IraniPruhs:Algorithmic-problems}, who
called it \emph{speed  scaling with  power  down scheduling}
problem.   We provide  an  $O(n^3)$  algorithm in  this  paper for  the
special case  of agreeable deadlines.  Jobs may be   released at different
time moments, and may have  distinct deadlines. The agreeable deadline
property  just  means  that   later  released  jobs  also  have  later
deadlines.  This  holds, for example, when  the deadline of  each job is
exactly $F$ units after  its release time,  which arises when one
wants to maintain a guarantee of service for the flow time of the jobs.

\section{Problem definition}

An  instance   of  our  scheduling  problem  consists   of  $n$  jobs,
$1,2,\ldots, n$, where each job $j,~ 1 \leq j \leq n$, is specified by
a release time/deadline interval\footnote{Notation: $[t_0,t_1)$ stands
  for the half  open interval $\{t:t_0\leq t <  t_1 \}$.} $[r_j, d_j)$
in which it  must be scheduled and a workload  $w_j$.  An instance has
the \emph{agreeable deadlines property}  if the jobs can be renumbered
such that both their release times and deadlines are in non-decreasing
order, i.e.  $i <  j$ implies $r_i  \leqslant r_j$ and  $d_i \leqslant
d_j$.

A schedule is defined by three functions 
\begin{eqnarray*}
  \text{mode}  & : & \mathbbm{R} \rightarrow \{\text{on},
  \text{off}\}\\
  \text{speed} & : & \mathbbm{R} \rightarrow \mathbbm{R^+}\\
  \text{job} & : & \mathbbm{R} \rightarrow \{\text{none}, 1, \ldots, n\},
\end{eqnarray*}

with the following properties

\begin{enumerate}
  \item $\forall t : \text{speed} (t) > 0 \Rightarrow \text{mode} (t) =
  \text{on}$

  \item $\forall t : \text{speed} (t) = 0 \Leftrightarrow \text{job} (t) =
  \text{none}$

  \item $\forall t : \text{job} (t) = j, j \neq \text{none} \Rightarrow t \in
  [r_j, d_j)$

  \item $\forall j\neq \text{none} : \int \text{speed} (t) \text{d} t = w_j$ where the integral
  is over all times $t$ such that $\text{job} (t) = j$

  \item for every time $t$, there is a positive length interval $I \ni t$ on
  which the schedule is constant. Moreover $I$ is of the form $(- \infty, u),
  [t', u)$ or $[t', + \infty)$ for some time points $t', u$.
  \label{piecewiseconstant}
\end{enumerate}

The last property is in fact a simplifying assumption to avoid degenerate schedules.
The    interpretation    is    that    at    a    time    $t$    where
$\text{job}(t)=\text{none}$, the machine is idle but switched on when
$\text{mode}(t)=\text{on}$      and     is     shut      down     when
$\text{mode}(t)=\text{off}$.   There   is   a  non-negligible   energy
consumption during the idle time periods, but one avoids the cost of
shutting down and rebooting the machine.

The cost (i.e., the consumed energy) of a schedule is specified by three parameters:
an exponent $\alpha \in [2, 3]$, a wake-up cost $L > 0$ and a ground
dissipation energy $g > 0$, and it has two components:

\begin{enumerate}
\item The \emph{speed cost}, that is the energy consumed in all times $t$ such that $\text{job} (t) = j \neq \text{none}$.
This cost is defined as $c_{\text{speed}} = \int_{} \text{speed} (t)^{\alpha} \text{d} t$.

\item The \emph{mode cost}, that is the cost of the ground dissipation
  energy plus the wake-up energy.
\end{enumerate}
A schedule with property (\ref{piecewiseconstant}) partitions the time into a sequence $S$
of disjoint, inclusion-wise maximal intervals, such that $\text{mode} (t) =
\text{on}$ if and only if $t \in \cup S=\cup_{I_k \in S} I$.
The sequence $S$ is called the \emph{support} of the schedule, and the
energy consumption generated by this support constitutes its  mode cost which is defined
as $c_{\text{mode}}= L \cdot (|S| + 1) + g \cdot | \cup S|$. Note that  we count a wake-up cost
$L$  for the two  half-infinite intervals  surrounding $S$.

Hence, the total cost is just the sum $c_\text{speed} + c_\text{mode}$
and the  problem studied in this  paper consists in  finding a minimum
cost schedule for an agreeable deadline instance.

The outline  of the  paper is  as follows. After  a brief  survey over
related work,  we start  showing  structural  properties of  optimal
schedules, in particular we introduce  the notion of prefix and suffix
of  a  block.   Finally  we  define the  dynamic  program,  prove  its
optimality and analyse its complexity, while mentioning implementation
issues. The  algorithm has been  implemented in Python and  can handle
instance of 300 jobs within a second.

\section{Previous work}

Our  scheduling  problem  for  general instances  (with  non-agreeable
deadlines)  was raised  in~\cite{IraniPruhs:Algorithmic-problems}.  No
polynomial time algorithm is known for this problem, nor has it 
been  shown to  be NP-hard.   The current  best positive  result  is a
$137/117$-approximation provided  in \cite{Abers.Antoniadis:Race}. The
same paper  also gives an  NP-hardness proof, however for  a different
energy   mode,   when   $c_{\text{speed}}$   is   defined   as   $\int
f(\text{speed} (t))  \text{d} t$ for some  piecewise linear monotone
function $f$. 

The general problem contains  two subproblems, which have been studied
and  solved  individually.  The  first  one  does  not consider  speed
scaling, and restricts to speeds 0  or 1, depending on the mode.  Here
essentially  the  goal  is  to  minimize  $c_\text{mode}$  only.  This
subproblem  has been solved  in $O(n^5)$  time by  dynamic programming
\cite{baptiste2007polynomial}. For  agreeable instances the complexity
has been  improved to $O(n^2)$ \cite{chau2011}.  The second subproblem
does  not consider  the power  down  mechanism, and  restricts to  the
single mode  'on' and  to ground dissipation  energy $g=0$.   Here the
problem is  to minimize $c_\text{speed}$ only.  This  problem has been
solved by a widely celebrated  greedy algorithm due to Yao, Demers and
Shenker \cite{yao1995scheduling} in  $O(n^3)$ time.  The complexity of
this algorithm, known  as YDS, has been improved  to $O(n^2\log n)$ in
\cite{li2006discrete}  and even  to $O(n^2)$  for  agreeable instances
\cite{wu2010min}.

Different  variants of  this problem  have  been studied  in the  past
years, which include the online setting, as well as different objective
values,  like   minimizing  throughput  or  flowtime.    We  refer  to
\cite{Albers:energy-efficient-algorithms} for an overview.

An important ingredient  for the algorithm presented in  this work, is
the aforementioned YDS algorithm.   For completeness we roughly sketch
it now.   At each step of  the algorithm, the  interval $I^{\star}$ of
maximum density is  selected among all $O(n^2)$ intervals  of the form
$[r_i,d_j]$. The \emph{density of an interval} is defined as the ratio
$s=W/(d_j-r_i)$, where $W$ denotes the  total workload of all jobs $k$
with  $[r_k,d_k]\subseteq  [r_i,d_j]$.   The  key  idea  is  that  any
feasible schedule must have for  $[r_i,d_j]$ an average speed at least
$s$.  Then all those  jobs are  scheduled in  $I^\star$ at  speed $s$
using the EDF (Earliest-Deadline-First)  policy.  No job will miss its
deadline  by maximal  density of  $I^\star$.   For the  sequel of  the
algorithm the  time interval $I^{\star}$ is  \emph{blacked out}.  This
means  that  when  computing  densities  of  candidate  intervals  for
subsequent iterations, the blacked out intervals are excluded, and the
schedule  for the  remaining  jobs  must exclude  them  as well.   The
algorithm ends when all jobs are scheduled.

\section{Structure of an optimal schedule}

When a job $j$ is running at speed $s$ its execution takes $w_j / s$ time units and
the consumed energy is $(s^{\alpha} + g) w_j / s$. This amount of energy is
minimum for speed $s^{\star} := (g / (\alpha - 1))^{1 / \alpha}$, which
we  call  the  \emph{critical  speed}.  Note  that  $s^\star$  is  job
independent. The \emph{density} of an interval $I$
is defined as $\sum w_j / |I|$ over all jobs $j$ with $[r_j, d_j) \subseteq I$. An interval is
called \emph{dense}, if its density is at least $s^{\star}$, and \emph{sparse}, otherwise.

\begin{lemma} \label{lem:structOpt}
  {\cite{IraniShuklaGupta}}  Given an instance of  the  speed scaling
  with  power down scheduling  problem, there  is an  optimal schedule
  $(\text{mode},\text{speed},\text{job})$ with the following properties.

  \begin{description}
    \item[job span] for every time $t$, if $\text{job} (t) = j \neq
    \text{none}$ then for all times $u \in [r_j, d_j)$ with $\text{mode} (u) =
    \text{on}$, we have $\text{speed} (u) \geqslant \text{speed} (t)$

    \item[earliest deadline first] for every time pair $t < u$ if $\text{job}
    (t) \neq \text{none} \neq \text{job} (u)$, then $\text{job} (t) \leqslant
    \text{job} (u)$.

    \item[dense intervals] dense intervals $I$ are scheduled according to the
    YDS rule.

  \item[domination] for any other optimal schedule
    $(\text{mode'},\text{speed'},\text{job'})$,  and  a smallest  time
    $t$  such   that  $\text{mode}(t)\neq  \text{mode'}(t)$   we  have
    $\text{mode}(t)=\text{on}$ and $\text{mode'}(t)=\text{off}$.
  \end{description}
\end{lemma}

In  particular  the  first   property implies  that  whenever  $j$  is
scheduled, the speed is the same. The next two properties imply
that dense intervals divide the problem into independent subproblems, as we
describe now.

\begin{definition}
  \label{def:subsinstance} A  subinstance of our  problem is specified
  by a pair $(i,j)$ with $i \in \{1, \ldots, n\}$, $j \in \{i-1,
  \ldots, n\}$. For convenience we denote  $d_0 = r_1 - L / g$ and $r_{n
    + 1} = d_n  + L / g$. It consists of the  interval $I= [d_{i - 1},
  r_{j + 1})$ and a job set $J$. If $i = j + 1$, then $J = \emptyset$,
  else $J  =\{i, \ldots, j\}$. The release  time/deadline intervals of
  these jobs are restricted by intersection to $I$.
\end{definition}

\begin{figure}[ht]
\def\svgwidth{14cm}

\begingroup
  \makeatletter
  \providecommand\color[2][]{%
    \errmessage{(Inkscape) Color is used for the text in Inkscape, but the package 'color.sty' is not loaded}
    \renewcommand\color[2][]{}%
  }
  \providecommand\transparent[1]{%
    \errmessage{(Inkscape) Transparency is used (non-zero) for the text in Inkscape, but the package 'transparent.sty' is not loaded}
    \renewcommand\transparent[1]{}%
  }
  \providecommand\rotatebox[2]{#2}
  \ifx\svgwidth\undefined
    \setlength{\unitlength}{994.4pt}
  \else
    \setlength{\unitlength}{\svgwidth}
  \fi
  \global\let\svgwidth\undefined
  \makeatother
  \begin{picture}(1,0.42777051)%
    \put(0,0){\includegraphics[width=\unitlength]{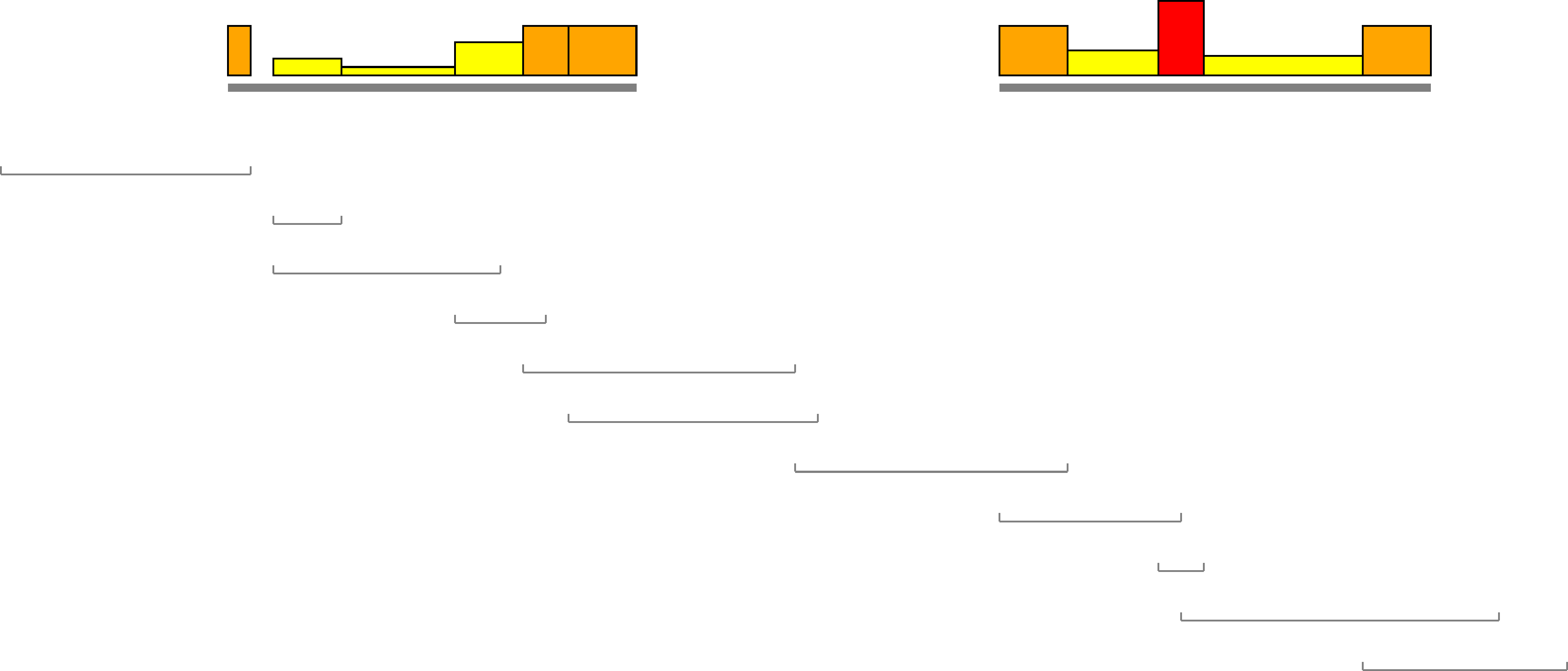}}%
    \put(0.00300667,0.32185653){\makebox(0,0)[lb]{\smash{w1=1}}}%
    \put(0.17671307,0.29025793){\makebox(0,0)[lb]{\smash{w2=1}}}%
    \put(0.17671307,0.25865934){\makebox(0,0)[lb]{\smash{w3=1}}}%
    \put(0.29251734,0.22706075){\makebox(0,0)[lb]{\smash{w4=2}}}%
    \put(0.33594395,0.19546216){\makebox(0,0)[lb]{\smash{w5=2}}}%
    \put(0.36489503,0.16386357){\makebox(0,0)[lb]{\smash{w6=3}}}%
    \put(0.50965033,0.13226497){\makebox(0,0)[lb]{\smash{w7=3}}}%
    \put(0.63993017,0.10066638){\makebox(0,0)[lb]{\smash{w8=2}}}%
    \put(0.74125892,0.06906778){\makebox(0,0)[lb]{\smash{w9=3}}}%
    \put(0.75573444,0.03746919){\makebox(0,0)[lb]{\smash{w10=3}}}%
    \put(0.87153872,0.00587061){\makebox(0,0)[lb]{\smash{w11=3}}}%
    \put(0.14776202,0.38505371){\makebox(0,0)[lb]{\smash{1}}}%
    \put(0.17671307,0.39558657){\makebox(0,0)[lb]{\smash{2}}}%
    \put(0.22013969,0.39032014){\makebox(0,0)[lb]{\smash{3}}}%
    \put(0.29251734,0.40611944){\makebox(0,0)[lb]{\smash{4}}}%
    \put(0.33594395,0.38505371){\makebox(0,0)[lb]{\smash{5}}}%
    \put(0.36489503,0.38505371){\makebox(0,0)[lb]{\smash{6}}}%
    \put(0.63993017,0.38505371){\makebox(0,0)[lb]{\smash{7}}}%
    \put(0.68335674,0.40085301){\makebox(0,0)[lb]{\smash{8}}}%
    \put(0.74125892,0.38505371){\makebox(0,0)[lb]{\smash{9}}}%
    \put(0.77020997,0.39734205){\makebox(0,0)[lb]{\smash{10}}}%
    \put(0.87153872,0.38505371){\makebox(0,0)[lb]{\smash{11}}}%
  \end{picture}%
\endgroup
\caption{Structure of an optimal schedule  for an instance of 11 jobs.
  The boxes represent job executions where the height equals speed and
  the area  equals the workload of  the job.  The  different colors of
  the boxes  distinguish critical speed, less than  critical speed and
  more than critical speed.  The thick line below represents the mode.
  Finally  the   intervals  at   the  bottom  represent   the  release
  time/deadline  intervals of each  job $j$,  labeled by  its workload
  $w_j$.   Here the  schedule  consists  of 2  blocks  separated by  a
  shutdown interval.  Jobs 1,5,6,7,11  are scheduled at critical speed
  $s^\star$,  while   job  9  is  scheduled  with   higher  speed,  as
  $[r_9,d_9)$ is a dense interval.  Note that the schedule is idle but
  not shutdown between jobs 1 and 2.} \label{fig:structOpt}
\end{figure}
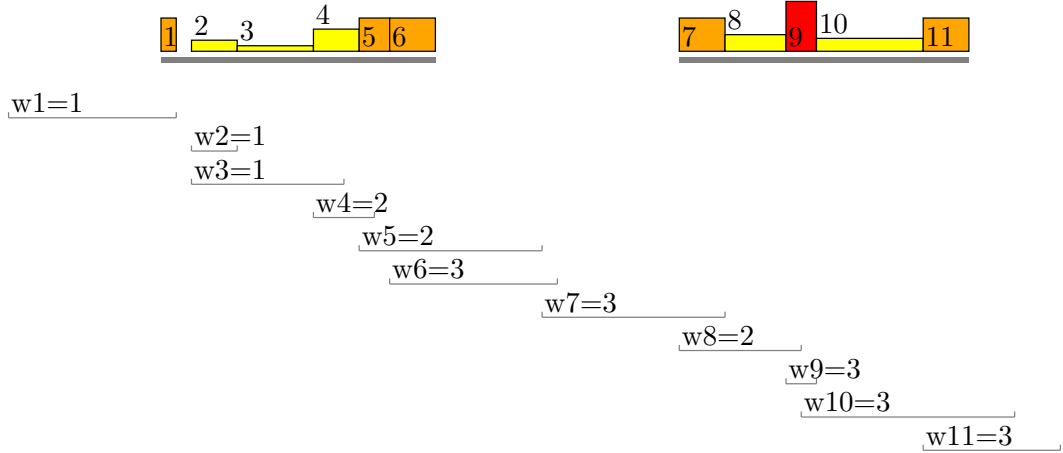

Note that  in case $d_{i -  1} < r_{j +  1}$ or $d_{i -  1} < d_i$
or $r_j < r_{j + 1}$, the subinstance $(i,j)$ is infeasible as the
release time/deadline interval of $i$ or $j$ is restricted
to the empty interval.

We extend also  the definition of the cost  function for subinstances.
The schedule of a subinstance $(i,  j)$, consisting of job set $J$ and
interval  $I$,   is  defined  by  the  functions   $\text{speed}  :  I
\rightarrow  \mathbbm{R^+}, \text{mode}  : I  \rightarrow \{\text{on},
\text{off}\}$ and $\text{job} : I \rightarrow \{\text{none}\} \cup J$.
For  the  mode  cost,  let  $S  :=  \{t \in  I  :  \text{mode}  (t)  =
\text{on}\}$ be the support of the  schedule, and $k$ be the number of
intervals in $I \backslash \cup S$. Then, $c_{\text{mode}} := k L + g|
\cup S|$.  The interpretation is  that if immediately before and after
$I$ the machine  is on, then shutdown intervals at  the borders of $I$
also do generate a wake-up cost.

We choose $d_0$ far enough from $r_1$ such that w.l.o.g. an optimal schedule
for the subinstance $(1,k)$ will  start  with  a shutdown  interval.  A  symmetric
property is true  for subinstances of the form  $(k,n)$. Therefore the
cost of the subinstance $(1,n)$ is consistent with the cost definition
for the complete instance.
Note that the optimum for a subinstance of the form $(i, i - 1)$ equals $\min
\{L, g (r_i - d_{i - 1})\}$.

Now consider all inclusion-wise maximal dense intervals. They partition the
time line into a sequence of alternating dense and sparse intervals.

The following  lemma follows directly from the  definitions. We stress
here that independence of the subschedules is implied by the agreeable
deadline assumption.   Lemma~\ref{lem:structOpt} states that there
is an  optimal schedule,  satisfying the \emph{earliest  deadline first
  property}, which means that whenever  job $j$ is scheduled, all jobs
$i<j$ already completed.  So the agreeable deadline assumption happens
to  be quite  strong, which  permits a  dynamic  programming approach.
However the problem does not become trivial, since one still needs to
decide when the machine is to be shutdown and when to be idle.
\begin{lemma}
  \label{lem:densePartition} Sparse intervals $I$ are associated to pairs $(i,
  j)$, such that the portion of an optimal schedule for the
  original instance restricted to $I$, is also an optimal schedule for the subinstance $(i, j)$.
  Moreover none of these subinstances contain dense intervals.
\end{lemma}

\section{Suffixes and  prefixes}

In  this section,  we consider  an  optimal schedule  of an  arbitrary
subinstance consisting of a job set  $J$ and an interval $I$ such that
all  subintervals of  $I$ are  sparse. Whenever,  in this  section, we
refer to release times/deadlines  $r_k, d_{\ell}$, they are restricted
to $I$.

\begin{lemma}
  \label{lem:belowCriticalSpeed}For all times $t \in I$, $\text{speed} (t)
  \leqslant s^{\star}$.
\end{lemma}

\begin{proof}
  Let  $t$ be  a time  that maximizes  $\text{speed} (t)$,  and assume
  $\text{speed}  (t) >  s^{\star}$ for the sake of contradiction. We
  consider an inclusion-wise  maximal interval $A \ni t$  on which the
  speed is constantly $\text{speed} (t)$. Let $i, \ldots,  k$,
  $i  \leqslant  \text{job} (t)  \leqslant   k$, be the jobs scheduled in this
  interval.  If $A =  [r_i, d_k)$, then  $A$ is a dense  interval, a
  contradiction to  Lemma  \ref{lem:densePartition}. Thus, the  inclusion  $A
  \subseteq [r_i,  d_k)$ is strict.  Assume $d_k >  u$ for $u  = \max
  A$ (the  other case  is symmetric). By  Lemma~\ref{lem:structOpt}, we
  have $\text{mode} (u)  = \text{off}$, and there is  a time $t'$ such
  that job $k$  is scheduled  in $[t',  u)$. For a  small enough
  $\delta \geqslant 1$, the execution of job $k$ can be extended to $[t',  u')$ for $u' =
  t' +  \delta (u -  t')$ and lower  its speed to $\text{speed}  (t) /
  \delta$. This   strictly decreases the  overall cost, a contradiction to the optimality
  of the schedule.
\end{proof}

The support of  the schedule consists of blocks  separated by shutdown
intervals. We shall show now that the boundaries  of  these  blocks  have a  particular
structure (see Figure~\ref{fig:prefix}).

\begin{definition}
  A suffix is  a job pair $(a,  b)$ such that all jobs  $a, \ldots, b$
  are scheduled at critical speed between  $r_a$ and $u$ with $u = r_a
  +  (w_a   +  \ldots  +  w_b)/s^{\star}$,  and  $\text{mode}  (u)  =
  \text{off}$. The definition of a prefix is just symmetric.
\end{definition}

\begin{figure}[ht]
\centerline{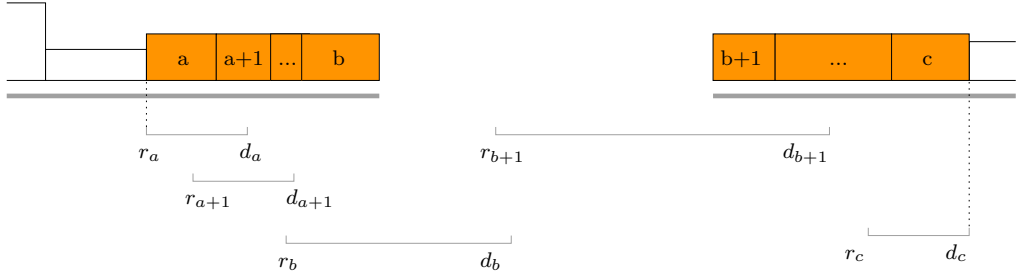}
  \caption{Illustration of a suffix  $(a,b)$ and a prefix $(b+1,c)$ in
    a schedule.  Note that job $a$  starts at its release time and job
    $c$ ends at its deadline.}
\label{fig:prefix}
\end{figure}

\begin{lemma}
  \label{lem:suffixPrefix}Let  $[t, u)$  be an  inclusion-wise maximal
  shutdown interval in $I$, that  is $\text{mode} (t')  = \text{off}$
  for all $t' \in [t, u)$. If  $t$ is not the start of $I$, then there
  is a  suffix $(a, b)$ ending  at $t = r_a  + (w_a + \ldots  + w_b) /
  s^{\star}$. If $u$ is not the end of $I$, then there is a prefix ($b
  +  1, c)$  starting at  $u  = d_c  - (w_{b  +  1}+ \ldots  + w_c)  /
  s^{\star}$.  Moreover, if both cases hold ($\inf I < t < u <\sup I$)
  then without loss of generality $r_{b + 1} > t$.
\end{lemma}

\begin{proof}
  Suppose that  there is an  execution interval $[t_0, t)$  where some
  job $b  = \text{job}  (t_0)$ is scheduled  at $\text{speed}  (t_0) <
  s^{\star}$. For a small enough $\delta>1$ let $t' := t_0 + (t - t_0)
  / \delta$. Consider  a new schedule where the  execution interval is
  \emph{compressed} to  $[t_0, t')$, the speed in  there is multiplied
  $\delta$, and the  shutdown interval is extended to  $[t', u)$. This
  new   schedule  has   a  strictly   decreased   cost,  contradicting
  optimality.

  This shows that if $t$ is not the start of $I$, then some job $b$ is
  scheduled right before $t$, say in some interval $[t_0, t)$, at critical
  speed. We will now show that there is a job $a$ such that between $r_a$ and
  $t$, jobs $a, \ldots, b$ are all scheduled at critical speed. If $t_0 = r_b$, we
  simply set $a = b$. Otherwise assume that $r_b < t_0$. If right before $t_0$ the
  schedule mode is off, then we can slightly shift the execution interval
  of $b$ to $[t_0 - \varepsilon, t - \varepsilon)$, to obtain a schedule
  of the same cost but with dominating work towards the beginning. W.l.o.g.
  we can assume that right before $t_0$ a job $b - 1$ is
  scheduled in some interval $[t_1, t_0)$. By the job span property of Lemma~\ref{lem:structOpt}
  and Lemma \ref{lem:belowCriticalSpeed} it is scheduled at speed $s^{\star}$. We
  iterate the arguments on $t_1$ and $b - 1$, eventually reaching a job $a$
  with the required property.

  The  same argument applied  symmetrically shows  the existence  of a
  prefix $(b + 1, c)$ if $u$ is not the end of $I$. Now if both suffix
  and prefix exist,  and $r_{b + 1} \leqslant t$,  then we could shift
  the execution of $b  + 1$ from $[u, w_{b + 1}  / s^{\star})$ to $[t,
  w_{b  +  1}  /  s^{\star})$,  yielding a  schedule  with  more  work
  dominating  towards the  beginning.  The  cost of  the  new schedule
  remains either the same or it is reduced by $L$, if $b+1$ were alone
  in its  block. Therefore  we can assume  w.l.o.g. that $t < r_{b + 1}$.
\end{proof}

To proceed to our dynamic programming algorithm we need one more property of
suffixes and prefixes implied by the following definition.

\begin{definition}
For a given subinstance $(i, j)$ we define two functions $f, h : \{i, \ldots,
j\} \rightarrow \{i, \ldots, j\}$ as follows: $f (a)$ is the highest index job
$b\leq j $ such that for all $a\leq k<b$, $r_a + (w_a + \ldots +
w_k) / s^{\star} \geqslant r_{k + 1}$, while $h (k)$ is the highest index job
$c\leq j$ such that for all $k<\ell \leq c\}$, $d_c - (w_{\ell} +
\ldots + w_c) / s^{\star} \leqslant d_{\ell - 1}$.
\end{definition}

\begin{lemma}
  \label{lem:fh}Any suffix $(a, b)$ satisfies $b = f (a)$ and any prefix $(k,
  c)$ satisfies $c = h (k)$.
\end{lemma}

The function $f$ requires a little more attention. Since by Lemma
\ref{lem:belowCriticalSpeed}, the job $a - 1$ cannot be scheduled with higher than
critical speed, we can assume that a suffix $(a, b)$ is such that $a$ is the
smallest index job with $f (a) = b$. So from now on we restrict the domain of
$f$ to those jobs. This allows $f$ to be invertible, i.e. $a = f^{- 1}
(b)$.  Note that by definition of $f$, the job $j$ is in the co-domain
of $f$, meaning that $f^{-1}(j)$ is defined.

\section{The dynamic program}

For every  subinstance $(i,  j)$, we denote  by $Y_{i, j}$  the minimum
$c_\text{speed}$ cost
plus  $g (r_{j  +  1}  - d_{i  -  1})$, and  by  $O_{i,j}$ the  minimum
$c_\text{speed}+c_\text{mode}$ cost.
If subinstance $(i,j)$ is infeasible we set $Y_{i,j},O_{i,j}$ to $+\infty$.
For convenience we denote $g^{\star} := (g + (s^{\star})^{\alpha}) / s^{\star}$.

\begin{figure}[ht]
\centerline{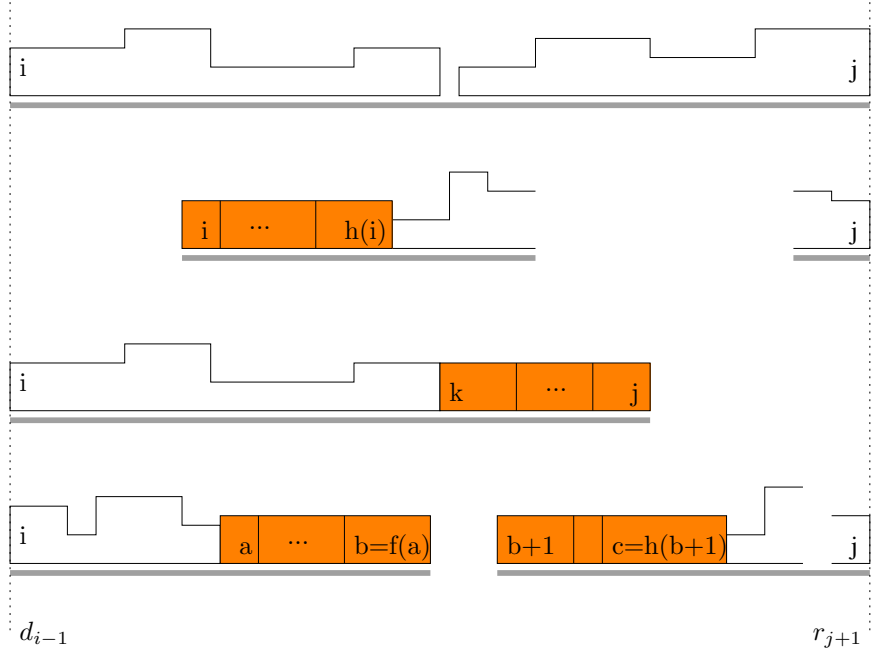}
  \caption{The four cases in (\ref{eq:defO}) of the dynamic program.}
\label{fig:fourCases}
\end{figure}

\begin{theorem}
  The value $O_{i,j}$ satisfies the following recursion.
  If $j=i-1$, then $O_{i,j}=min\{L, g(r_{j+1}-d_{i-1})\}$, otherwise, let $k=f^{-1}(j)$.

\begin{equation}
\label{eq:defO}
  O_{i,j} = \min \left\{ \begin{array}{l}
       Y_{i,j}\\
         L + g^{\star} (w_i + \ldots + w_{h (i)}) + O_{h (i) + 1, j}\\
       Y_{i, k - 1} + g^{\star} (w_k + \ldots + w_j) + L\\
       \min Y_{i,a - 1} + g^{\star} (w_a + \ldots + w_b) + L + g^{\star} (w_{b +
       1} + \ldots + w_c) + O_{c + 1, j},
     \end{array} \right.
\end{equation}
  where the inner minimization is over all jobs $a\in\{i+1, \ldots j\}\cup\{b,c\}$ with $b = f(a), b < j$ and $c = h (b + 1)$.
  As usual if there are no such jobs, the value of this inner minimization is $+ \infty$.
\end{theorem}


\begin{proof}
The case $j=i-1$ is simple, since the optimal empty schedule is either
idle or shutdown depending on the span of $[r_{j+1},d_{i-1})$.

Now  for  some  $i\leq  j$,  consider the  subinstance  $(i,  j)$.  By
induction on  $j-i$, we can  show that for  each of the four  cases in
(\ref{eq:defO}) there is a feasible schedule with the corresponding cost.  For the
remainder of the proof, we consider a schedule $S$ minimizing
$c_\text{speed}+c_\text{mode}$ for this subinstance, and we show that one
of the four cases yields its cost.

If $S$ is never power down, then the contribution of $c_{\text{mode}}$ is
exactly  $g  (r_{j  + 1}  -  d_{i  -  1})$,  and the  contribution  of
$c_{\text{speed}}$ is minimal. So the first case applies.

Now suppose that there is some interval $[t, u)$ where the schedule is
power down,  $[t, u)$  is inclusion-wise  maximal and it is the first interval.
There are several cases now, depending on the conditions $t =  \min I$, $u = \max I$, where  $I$ is the interval
associated to the sub-instance $(i,j)$.

It cannot be that both conditions  are true, since this means that the
schedule is empty, which contradicts the case assumption $i\leq j$.

If  $t = \min  I$ and  $u< \max I$,  then by  Lemma \ref{lem:suffixPrefix}
there is a prefix $(i, c)$ of the form $c = h (i)$. The portion up to
$d_c$ of this  schedule has a contribution to cost equal to $L + g^{\star}  (w_i + \ldots +
w_{c})$,   and  by  the  composition  of  schedules,   its   remainder
has a contribution of  $O_{c +1,j}$.   Hence, the second case  of
(\ref{eq:defO}) applies.

If $t > \min I$ and $u = \max I$, similarly there is a suffix $(k, j)$
and the cost of the schedule up  to $r_k$ is $Y_{i,k-1}$, since there are
no power down states, while the remainder contributes a cost of  $g^{\star} (w_k +
\ldots  +  w_j)   +  L$.   This  time  it  is   the  third  case  of
(\ref{eq:defO}) which applies.

If   $t   >   \min  I$   and   $u   <   \max   I$,  again   by   Lemma
\ref{lem:suffixPrefix}, there is a suffix $(a, b)$ and a prefix $(b_{}
+ 1, c)$ around a power down interval  $[t, u)$, and by Lemma \ref{lem:fh} we
have $b = f  (a), c = h (b + 1)$. Then, the   cost of the schedule decomposes
into the cost $Y_{i,  a-1}$ for the part  before $r_a$, since it  does not contain
power down states,  the cost $g^{\star} (w_a +  \ldots + w_c) +  L$ for the
part in $[r_a, d_c)$, and the cost  $O_{c+1, j}$ for the remainder, by the composition
of schedules. In  this final case, the last  case of (\ref{eq:defO})
applies.
\end{proof}

\section{Complexity analysis}

The dynamic program uses $O(n^2)$ variables, and for each one of them a minimization over
$O(n)$ values is required. Therefore, it can be run in $O (n^3)$ time.

For a fixed subinstance $(i, j)$ the functions $f, h$ can be computed by
simple scanning procedures in linear time as following (we omit their proof of correctness).

\begin{itemize}
  \item Initially $\ell := i$ and $t := r_i$. For all $k = i, i + 1,
  \ldots, j$, if $t < r_k$, then $\ell := k$, $t := r_k .$ In any
  case $f (\ell) := k$, $t := t + w_k / s^{\star}$.

  \item Initially $\ell := j$ and $t := d_j$. For all $k = j, j - 1,
  \ldots, i$, if $t > d_k$, then $\ell := k$, $t := d_k$. In any
  case $h (k) = \ell$, $t := t - w_k / s^{\star}$.
\end{itemize}

The computation of the values $Y_{i, j}$ however is crucial, there are $O
(n^2)$ of them and the best known algorithm to compute the optimal schedule
for each of  them runs in time $O  (n^2)$ \cite{li2006discrete}, which
would lead  to a total  running time of  $O (n^4)$. We now  describe a
procedure which  permits to compute $Y_{i,j}$ iteratively  from $Y_{i -
  1, j}$ in total time $O(n^2)$. Therefore we can compute all
optimal $c_\text{speed}$ subschedules in total time $O (n^3)$.

\subsection{Computing $Y_{i, j}$}

The general outline is as follows. We first compute
$Y_{1,   n}$    in   time   $O(n^2)$   using    the   algorithm   from
\cite{li2006discrete}. Then in  a first right to left  scan we compute
all values  $Y_{1, j}$ for $j =  n-1,\ldots, 1$.  After  that for every
$j$, we  apply a left to right  scan to compute all  values $Y_{i, j}$
for $i = 2, \ldots, j$. This left to right scan works as follows.

It receives as input the $c_\text{speed}$-optimal schedule $S$ for the
subinstance  $(1,  j)$,  and  applies the  following  \emph{squeezing}
procedure  to  $S$.   The  schedule  $S$ consists  of  a  sequence  of
\emph{blocks},  every block  spans  some time  interval  $[t, u)$  and
contains  a  sequence of  jobs  running  at  some constant,  but  block
dependent speed.

During the  procedure we keep track  of the first block which spans time interval
$[t,  u)$ and schedules the jobs $i,  \ldots,  b$ at  speed $s$.   Initially
$i=1$.  We consider the action  of squeezing the block to the interval
$[u   -    \ell,   u)$   by   increasing   the    speed   $s$,   where
$\ell:=u-(w_i+\ldots+w_b)/s$.

While $i\leq j$,  we decide which of the following events happens first, and
execute the corresponding actions.
\begin{description}
\item[unfeasibility event:]
  It  happens  when   $d_{i-1}=d_i$  or
  $r_j=r_{j-1}$.  Since  in  the  subinstance $(i,j)$  all  jobs  are
  restricted to the interval $[d_{i-1},r_{j+1}]$, it follows that one of
  the jobs  $i,j$ is  restricted to an  empty interval, and  cannot be
  scheduled  with  finite  speed.  In  this  case,  we  announce  that
  subinstance $(i,j)$ is  unfeasible, we remove job $i$  from $S$, and
  increase $i$.

\item[merge event:] It happens when the current speed
  $s$ equals $\text{speed} (u)$.  In this case we merge the first two blocks.
  (Note that  if $u=d_b$,  then this event  will immediately  be followed  by 
  the next split event for the merged block.)

\item[split event:] At  some moment, a job $i\leq  k<b$ from the first
  block might  complete at its  deadline. This happens when  the speed
  $s$ reaches $\hat  s(k,b,u) := (w_{k+1}+\ldots+w_b)/(u-d_k)$. In this
  case the  block splits into  two new blocks  with the first  of them
  restricted to  the interval $[t, d_k)$  and to the  jobs $i, \ldots,
  k$.

\item[deadline   event:]  When   $s=(w_i+\ldots+w_b)/(u-d_{i-1})$,  the
  current  schedule $S$ is  the optimal  $c_\text{speed}$-schedule for
  the  subinstance $(i,  j)$. In this case we output  $S$ as  $Y_{i,  j}$, we
  remove job $i$ from $S$, and increase $i$.
\end{description}

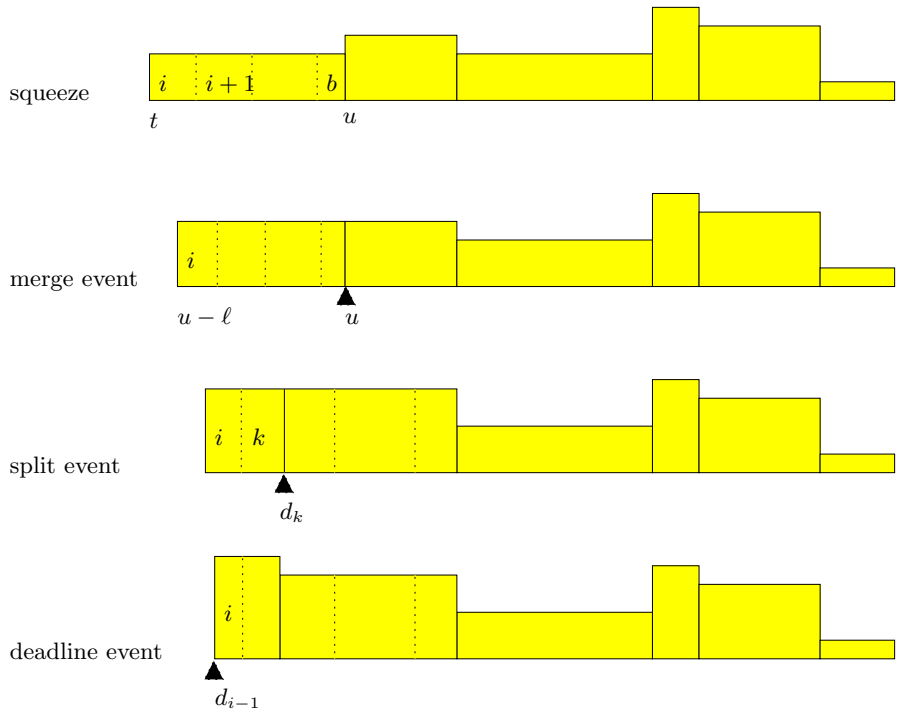
\begin{figure}[ht]
\begin{picture}(0,0)%
\includegraphics{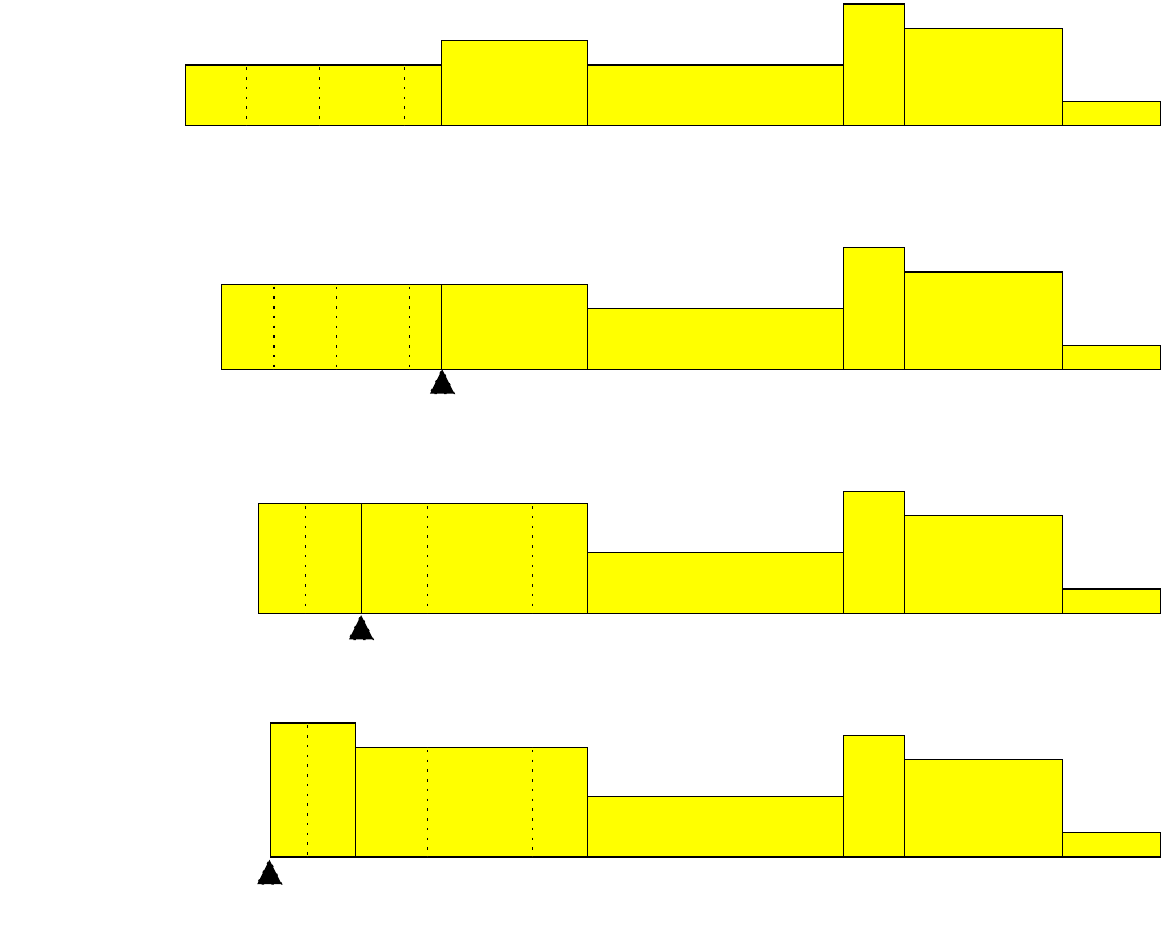}%
\end{picture}%
\setlength{\unitlength}{3079sp}%
\begingroup\makeatletter\ifx\SetFigFont\undefined%
\gdef\SetFigFont#1#2#3#4#5{%
  \reset@font\fontsize{#1}{#2pt}%
  \fontfamily{#3}\fontseries{#4}\fontshape{#5}%
  \selectfont}%
\fi\endgroup%
\begin{picture}(7152,5715)(1486,-7414)
\put(4181,-2646){\makebox(0,0)[lb]{\smash{{\SetFigFont{9}{10.8}{\rmdefault}{\mddefault}{\updefault}{\color[rgb]{0,0,0}$u$}%
}}}}
\put(3676,-5836){\makebox(0,0)[lb]{\smash{{\SetFigFont{9}{10.8}{\familydefault}{\mddefault}{\updefault}{\color[rgb]{0,0,0}$d_k$}%
}}}}
\put(4201,-4261){\makebox(0,0)[lb]{\smash{{\SetFigFont{9}{10.8}{\rmdefault}{\mddefault}{\updefault}{\color[rgb]{0,0,0}$u$}%
}}}}
\put(1501,-6961){\makebox(0,0)[lb]{\smash{{\SetFigFont{9}{10.8}{\familydefault}{\mddefault}{\updefault}{\color[rgb]{0,0,0}deadline event}%
}}}}
\put(1501,-5461){\makebox(0,0)[lb]{\smash{{\SetFigFont{9}{10.8}{\familydefault}{\mddefault}{\updefault}{\color[rgb]{0,0,0}split event}%
}}}}
\put(1501,-3961){\makebox(0,0)[lb]{\smash{{\SetFigFont{9}{10.8}{\familydefault}{\mddefault}{\updefault}{\color[rgb]{0,0,0}merge event}%
}}}}
\put(2701,-2386){\makebox(0,0)[lb]{\smash{{\SetFigFont{9}{10.8}{\familydefault}{\mddefault}{\updefault}{\color[rgb]{0,0,0}$i$}%
}}}}
\put(2626,-2686){\makebox(0,0)[lb]{\smash{{\SetFigFont{9}{10.8}{\rmdefault}{\mddefault}{\updefault}{\color[rgb]{0,0,0}$t$}%
}}}}
\put(3451,-5236){\makebox(0,0)[lb]{\smash{{\SetFigFont{9}{10.8}{\familydefault}{\mddefault}{\updefault}{\color[rgb]{0,0,0}$k$}%
}}}}
\put(3226,-6661){\makebox(0,0)[lb]{\smash{{\SetFigFont{9}{10.8}{\familydefault}{\mddefault}{\updefault}{\color[rgb]{0,0,0}$i$}%
}}}}
\put(3151,-5236){\makebox(0,0)[lb]{\smash{{\SetFigFont{9}{10.8}{\familydefault}{\mddefault}{\updefault}{\color[rgb]{0,0,0}$i$}%
}}}}
\put(2926,-3811){\makebox(0,0)[lb]{\smash{{\SetFigFont{9}{10.8}{\familydefault}{\mddefault}{\updefault}{\color[rgb]{0,0,0}$i$}%
}}}}
\put(3076,-2386){\makebox(0,0)[lb]{\smash{{\SetFigFont{9}{10.8}{\familydefault}{\mddefault}{\updefault}{\color[rgb]{0,0,0}$i+1$}%
}}}}
\put(1501,-2461){\makebox(0,0)[lb]{\smash{{\SetFigFont{9}{10.8}{\familydefault}{\mddefault}{\updefault}{\color[rgb]{0,0,0}squeeze}%
}}}}
\put(2851,-4261){\makebox(0,0)[lb]{\smash{{\SetFigFont{9}{10.8}{\rmdefault}{\mddefault}{\updefault}{\color[rgb]{0,0,0}$u-\ell$}%
}}}}
\put(4051,-2386){\makebox(0,0)[lb]{\smash{{\SetFigFont{9}{10.8}{\familydefault}{\mddefault}{\updefault}{\color[rgb]{0,0,0}$b$}%
}}}}
\put(3151,-7336){\makebox(0,0)[lb]{\smash{{\SetFigFont{9}{10.8}{\familydefault}{\mddefault}{\updefault}{\color[rgb]{0,0,0}$d_{i-1}$}%
}}}}
\end{picture}%
\caption{Different events during the squeeze procedure}
\label{fig:squeeze}
\end{figure}

At  any  moment  the  algorithm  maintains  a  schedule  $S$  for  the
subinstance consisting  of all  jobs $i,\ldots,j$ with  release times and
deadlines restricted to the  interval $[u-\ell, r_{j+1}]$. We omit the
proof of optimality of $S$ which should be straightforward.

It remains to specify how the next event can be determined in constant
time. The  merge and deadline events,  are both specified  by a single
expression determining the value $\ell$  at which they occur. For the
split  event the  situation  is  more subtle,  since  there are  $b-i$
candidates $\hat s(k,b,u)$, one for each job $i\leq k<b$.  We handle this
by precomputing  $\hat s$. Note  that for a  given job $b$,  there are
only $O(n)$ different times $u$ to be considered, and they are of the form $d_b,r_{b+1}$
and $r_{j+1}$ for all $1\leq j \leq n$.  This is because every block of an optimal
schedule ends either at the end of the interval $I$ if it
is the  last block, or at  one of $d_b,r_{b+1}$,  depending on whether
the next block has lower or higher speed.

This means that there are  $O(n^3)$ values of the form $\hat s(k,b,u)$
to compute, and  this can be done for each pair  $b,u$ in linear time,
by iterating $k$ from $b-1$ to $1$.  In the procedure above we need to
determine the job $k,~ i\leq k<b$ minimizing $\hat s(k,b,u)$.
Clearly, such a job $k$ can  be computed in constant
time for each triplet $(i,b,u)$,  again by iterating $i$ from $b-1$ to
$1$ for each pair $b,u$.

In  the event  loop described  above every  job  is responsible  for  at most three
events. Therefore its complexity is $O(n)$ for fixed $j$, which yields to a
total running time of  $O(n^3)$.

\section{Conclusion}

We provided  a polynomial  time algorithm for  the speed  scaling with
power down scheduling  problem,  for the  special  case of  agreeable
deadlines.  This assumption leads to strong structural properties of optimal schedules,
which are non-preempted  and, moreover, permit a partitioning
leading  to a dynamic programming algorithm.  So the  proposed algorithm  could not
be generalized to instances with arbitrary deadlines. However, we believe that the squeezing
procedure could be of independent interest.

\bibliographystyle{elsarticle-num}
\bibliography{4S}

\end{document}